\documentclass[preprint]{elsarticle}
\usepackage{amsthm,amsmath,amssymb}
\usepackage{algorithm,algpseudocode}
\usepackage[unicode]{hyperref}
\usepackage{booktabs}
\usepackage{comment}
\usepackage{tikz}
\usetikzlibrary{positioning}

\makeatletter
\def\ps@pprintTitle{%
 \let\@oddhead\@empty
 \let\@evenhead\@empty
 \def\@oddfoot{\centerline{\thepage}}%
 \let\@evenfoot\@oddfoot}
\makeatother

\newtheorem{theorem}{Theorem}
\newtheorem{lemma}[theorem]{Lemma}

\usepackage{color}
\usepackage{CJKutf8}
\usepackage{ifthen}
\newcommand{\COMM}[2]{{
\begin{CJK}{UTF8}{ipxm}
\ifthenelse{\equal{#1}{SK}}{\color{blue}}{
\ifthenelse{\equal{#1}{TM}}{\color{red}}{
\ifthenelse{\equal{#1}{AA}}{\color{cyan}}{
\ifthenelse{\equal{#1}{BB}}{\color{magenta}}}}}
[#1: #2]
\end{CJK}
}}

\begin{document}

\begin{frontmatter}
\title{Incorrect implementations of the Floyd--Warshall algorithm give correct solutions after three repeats}

\author{Ikumi Hide}
\address{The University of Tokyo \\ ihide@es.a.u-tokyo.ac.jp}

\author{Soh Kumabe}
\address{The University of Tokyo \\ sohkuma0213@gmail.com}

\author{Takanori Maehara}
\address{RIKEN Center for Advanced Intelligence Project \\ takanori.maehara@riken.jp} 
%https://www.overleaf.com/project/5c62987bcf1a1158e925fe73

\begin{abstract}
The Floyd--Warshall algorithm is a well-known algorithm for the all-pairs shortest path problem that is simply implemented by triply nested loops.
In this study, we show that the incorrect implementations of the Floyd--Warshall algorithm that misorder the triply nested loops give correct solutions if these are repeated three times.
\end{abstract}


\end{frontmatter}

\section{Introduction}

The \emph{Floyd--Warshall algorithm} is a well-known algorithm for the all-pairs shortest path problem~\cite{cormen2009introduction,floyd1962algorithm}.
Let $G = (V, E)$ be a complete directed graph, where $V = \{1, \dots, n\}$ is the set of vertices, and let $w \colon V \times V \to \mathbb{R} \cup \{\infty\}$ be the length of the edges.
We assume that $G$ has no negative cycles.
The Floyd--Warshall algorithm maintains array $\texttt{d}$ of size $n \times n$ initialized by $\texttt{d}[i, j] \leftarrow w(i, j)$ for all $i, j \in V$, and performs triply nested loops as shown in Algorithm~\ref{alg:kij}. 
Then, $\texttt{d}[i, j]$ eventually stores the shortest path distances for all $i, j \in V$.

\begin{algorithm}[b]
\caption{The Floyd--Warshall Algorithm}
\label{alg:kij}
\begin{algorithmic}[1]
\For{$k = 1, 2, \dots, n$}
\For{$i = 1, 2, \dots, n$}
\For{$j = 1, 2, \dots, n$}
\State{$\texttt{d}[i,j] \leftarrow \min \{ \texttt{d}[i,j], \texttt{d}[i,k] + \texttt{d}[k,j]\} $}
\EndFor
\EndFor
\EndFor
\end{algorithmic}
\end{algorithm}

A common mistake in implementing the Floyd--Warshall algorithm is to misorder the triply nested loops.\footnote{e.g., https://cs.stackexchange.com/questions/9636/why-doesnt-the-floyd-warshall-algorithm-work-if-i-put-k-in-the-innermost-loop}
The correct order is KIJ, and the incorrect versions, which are referred as \emph{IJK algorithm} and \emph{IKJ algorithm}, are shown in Algorithm~\ref{alg:ijk} and Algorithm~\ref{alg:ikj}, respectively.
These incorrect versions do not give correct solutions for some instance. 
However, we can prove that if these are repeated three times, we obtain the correct solutions.
More precisely, we obtain the following theorems.
\begin{theorem}
\label{thm:ijk}
If we repeat IJK algorithm three times, it solves the all-pairs shortest path problem.
Conversely, there exists an instance that needs three repeats to obtain a correct solution.
\end{theorem}
\begin{theorem}
\label{thm:ikj}
If we repeat IKJ algorithm two times, it solves the all-pairs shortest path problem.
Conversely, there exists an instance that needs two repeats to obtain a correct solution.
\end{theorem}
It would be emphasized that these fixes (repeating incorrect algorithms three times) have the same time complexity as the correct Floyd--Warshall algorithm up to constant factors.
Therefore, our results suggest that, if one is confused by the order of the triply nested loops, one can repeat the procedure three times just to be safe.

\begin{algorithm}[tb]
\caption{IJK algorithm}
\label{alg:ijk}
\begin{algorithmic}[1]
\For{$i = 1, 2, \dots, n$}
\For{$j = 1, 2, \dots, n$}
\For{$k = 1, 2, \dots, n$}
\State{$\texttt{d}[i,j] \leftarrow \min \{ \texttt{d}[i,j], \texttt{d}[i,k] + \texttt{d}[k,j]\} $}
\EndFor
\EndFor
\EndFor
\end{algorithmic}
\end{algorithm}

\begin{algorithm}[tb]
\caption{IKJ algorithm}
\label{alg:ikj}
\begin{algorithmic}[1]
\For{$i = 1, 2, \dots, n$}
\For{$k = 1, 2, \dots, n$}
\For{$j = 1, 2, \dots, n$}
\State{$\texttt{d}[i,j] \leftarrow \min \{ \texttt{d}[i,j], \texttt{d}[i,k] + \texttt{d}[k,j]\} $}
\EndFor
\EndFor
\EndFor
\end{algorithmic}
\end{algorithm}

\section{Proofs}

Both theorems are proved in a similar way. 
The converse parts (existence of bad instances) are proved by the exhaustive search, and such instances are shown in Figures~\ref{fig:ijk} and \ref{fig:ikj}.

To prove the if part, we regard the algorithms as graph modification processes.
Initially, we have the complete directed graph with edge length $\texttt{d}$ initialized by $w$.
Each update $\texttt{d}[i,j] \leftarrow \min \{ \texttt{d}[i,j], \texttt{d}[i,k] + \texttt{d}[k,j] \}$ modifies the edge length from $i$ to $j$ by the length of the path $[i, k, j]$ if it is shorter than the current length.
It is important that this modification does not change the shortest path distances.

In the following, we fix arbitrary vertices $s, t \in V$ and prove that, after several run of the algorithms, $\texttt{d}[s, t]$ stores the shortest path distance from $s$ to $t$. 
Since $s$ and $t$ are arbitrary, this proves the theorems.
We take arbitrary shortest path $[u_0, u_1, \dots, u_l]$ from $s$ to $t$ (i.e., $u_0 = s$ and $u_l = t$) and analyze the structure of the shortest path after single run of the algorithms.

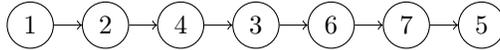
\begin{figure}[tb]
    \centering
    \begin{tikzpicture}
    \node[circle,draw] at (0,0) (a) {1};
    \node[circle,draw] at (1,0) (b) {2};
    \node[circle,draw] at (2,0) (c) {4};
    \node[circle,draw] at (3,0) (d) {3};
    \node[circle,draw] at (4,0) (e) {6};
    \node[circle,draw] at (5,0) (f) {7};
    \node[circle,draw] at (6,0) (g) {5};
    \draw[->] (a)--(b);
    \draw[->] (b)--(c);
    \draw[->] (c)--(d);
    \draw[->] (d)--(e);
    \draw[->] (e)--(f);
    \draw[->] (f)--(g);
    \end{tikzpicture}
    \caption{A smallest instance that needs three repeats of the IJK algorithm. All the edge lengths are one.}
    \label{fig:ijk}
\end{figure}

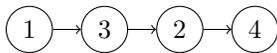
\begin{figure}[tb]
    \centering
    \begin{tikzpicture}
    \node[circle,draw] at (0,0) (a) {1};
    \node[circle,draw] at (1,0) (b) {3};
    \node[circle,draw] at (2,0) (c) {2};
    \node[circle,draw] at (3,0) (d) {4};
    \draw[->] (a)--(b);
    \draw[->] (b)--(c);
    \draw[->] (c)--(d);
    \end{tikzpicture}
    \caption{A smallest instance that needs two repeats of the IKJ algorithm. All the edge lengths are one.}
    \label{fig:ikj}
\end{figure}

We first prove Theorem~\ref{thm:ijk}.
We use the following lemma.
For integers $a$ and $c$ with $a < c$, we denote by $[a, c] = \{a, a+1, \dots, c\}$.
\begin{lemma}
\label{lem:dent}
Let $[a, c]$ be an interval such that $u_a > u_b$ and $u_c > u_b$ for all $b \in [a+1, c-1]$.
After the loop at $i = u_a$, $\texttt{d}[u_a, u_c]$ stores the shortest path distance from $u_a$ to $u_c$.
\end{lemma}
\begin{proof}
We prove this lemma by the induction on the length of the interval $[a, c]$.
If the length is at most three, the lemma immediately holds.
In general case, let $p \in [a, c]$ be such that $u_p$ is the highest in $\{ u_{a+1}, \dots, u_{c-1} \}$.
Then, $[a, p]$ and $[p, c]$ satisfy the condition of the lemma.
By the induction hypothesis, before the loop at $i = u_a$, 
$\texttt{d}[u_p, u_c]$ stores the shortest path distance from $u_p$ to $u_c$.
Also, by the induction hypothesis, after the loop at $i = u_a$, $\texttt{d}[u_a, u_p]$ stores the shortest path distance from $u_a$ to $u_p$.
Here, since $u_p < u_c$, the update of $\texttt{d}[u_a, u_c]$ is performed after the update of $\texttt{d}[u_a, u_p]$.
Therefore, by the update $\texttt{d}[u_a, u_c] \leftarrow \min \{\texttt{d}[u_a, u_c], \texttt{d}[u_a, u_p] + \texttt{d}[u_p, u_c] \}$, $\texttt{d}[u_a, u_c]$ stores the shortest path distance from $u_a$ to $u_c$.
\end{proof}

\begin{proof}[Proof of Theorem~1]
By Lemma~\ref{lem:dent}, after the first run of the IJK algorithm, the modified graph has a shortest path $[v_0, \dots, v_{l'}]$ whose vertices are upper unimodal, i.e., $s = v_0  < v_1 < \dots < v_h$ and $v_h > \dots > v_{l'-1} > v_{l'} = t$.
After the second run of the IJK algorithm,
$\texttt{d}[s, v_h]$ stores the shortest path distance from $s$ to $v_h$ because of the sequential update $\texttt{d}[s, j] \leftarrow \min \{ \texttt{d}[s, j], \texttt{d}[s, k] + \texttt{d}[k, j] \}$ for $(j, k) = (v_2, v_1), (v_3, v_2), \dots, (v_{h}, v_{h-1})$.
Also, $\texttt{d}[v_h, t]$ stores the shortest path distance from $v_h$ to $t$ because of the sequential update $\texttt{d}[i, t] \leftarrow \min \{ \texttt{d}[i, t], \texttt{d}[i, k] + \texttt{d}[k, t] \}$ for $(i, k) = (v_{l'-2}, v_{l'-1}), (v_{l'-3}, v_{l'-2}), \dots, (v_h, v_{h+1})$.
Therefore, after the third run of the IJK algorithm, $\texttt{d}[s, t]$ stores the shortest path distance from $s$ to $t$ because of the update $\texttt{d}[s, t] \leftarrow \min \{ \texttt{d}[s, t], \texttt{d}[s, v_p] + \texttt{d}[v_p, t] \}$
\end{proof}

We then prove Theorem~\ref{thm:ikj}.
We use the following lemma.
\begin{lemma}
\label{lem:descent}
Let $[a, c]$ be an interval such that $u_a > u_b$ for all $b \in [a+1, c-1]$.
After the loop for $i = u_a$, $\texttt{d}[u_a, u_c]$ stores the shortest path distance from $u_a$ to $u_c$.
\end{lemma}
\begin{proof}
We prove this by the induction on the length of the interval $[a, c]$.
If the length is at most three, the lemma immediately holds.
In general case, we divide the interval $[a+1, c]$ into $[a_1, a_2]$, $[a_2, a_3]$, $\dots$, $[a_l, a_{l+1}]$ (i.e., $a_1 = a+1$ and $a_{l+1} = c$) such that each interval satisfies the condition of the lemma.
Here, we observe that $u_{a_1} < u_{a_2} < \dots < u_{a_l} < u_a$.
By the induction hypothesis, before the loop at $i = u_a$, all $\texttt{d}[u_{a_p}, u_{a_{p+1}}]$ ($p = 1, \dots, l$) store the shortest path distances.
Hence, after the loop at $i = u_a$, by the sequential update $\texttt{d}[u_a, u_{a_{p+1}}] \leftarrow \min \{ \texttt{d}[u_a, u_{p+1}], \texttt{d}[u_a, u_{a_p}] + \texttt{d}[u_{a_p}, u_{a_{p+1}}] \}$ for $p = 1, \dots, l$, $\texttt{d}[u_a, u_c]$ stores the shortest path distance.
\end{proof}

\begin{proof}[Proof of Theorem~2]
By Lemma~\ref{lem:descent}, after the first run of the IKJ algorithm, the modified graph has a shortest path $[v_0, \dots, v_{l'}]$ whose vertices are monotonically increasing except the last one, i.e., $s = v_0 < v_1 < \dots < v_{l'-1}$.
After the second run of the IKJ algorithm, $\texttt{d}[s, t]$ stores the shortest path distance from $s$ to $t$ because of the sequential update $\texttt{d}[s, j] \leftarrow \min \{ \texttt{d}[s, j], \texttt{d}[s, k] + \texttt{d}[k, j] \}$ for $(k, j) = (v_1, v_2), (v_2, v_3), \dots, (v_{l'-1}, v_{l'})$.
\end{proof}

\bibliographystyle{plain}
\bibliography{main}

\end{document}